\documentclass[preprint,authoryear,12pt]{elsarticle}

\usepackage{graphicx}

\usepackage{amssymb}
\usepackage{amsthm}

\usepackage{lineno}

\usepackage{amsmath,accents}
\usepackage{amsfonts}
\usepackage{bbm}
\usepackage{bm}
\usepackage{mathrsfs}
\usepackage{algorithmic}
\usepackage{algorithm}
\usepackage{booktabs}
\usepackage{multirow}
\usepackage{enumerate}
\usepackage{setspace}

\usepackage{lscape}

\usepackage[table]{xcolor}
\usepackage{booktabs}

\usepackage{array}

\newcolumntype{x}[1]{%
>{\raggedleft\hspace{0pt}}p{#1}}%

\DeclareFontFamily{U}{mathx}{\hyphenchar\font45}
\DeclareFontShape{U}{mathx}{m}{n}{
      <5> <6> <7> <8> <9> <10>
      <10.95> <12> <14.4> <17.28> <20.74> <24.88>
      mathx10
      }{}
\DeclareSymbolFont{mathx}{U}{mathx}{m}{n}
\DeclareFontSubstitution{U}{mathx}{m}{n}
\DeclareMathAccent{\widecheck}{0}{mathx}{"71}
\DeclareMathAccent{\wideparen}{0}{mathx}{"75}

\newcommand{\interior}[1]{\accentset{\smash{\raisebox{-0.12ex}{$\scriptstyle\circ$}}}{#1}\rule{0pt}{2.3ex}}
\newcommand{\eps}{\varepsilon}

\newcommand{\E}{\mathsf{E}}
\newcommand{\var}{\mathsf{Var}\,}

\renewcommand{\P}{\mathsf{P}}

\newcommand{\zero}{\mathbf{0}}

\newcommand{\balpha}{\bm{\alpha}}
\newcommand{\bbeta}{\bm{\beta}}
\newcommand{\bgamma}{\bm{\gamma}}

\newcommand{\btheta}{\bm{\Theta}}

\newtheorem{theorem}{Theorem}

\newtheorem{corollary}[theorem]{Corollary}

\theoremstyle{definition}
\newtheorem{definition}{Definition}

\theoremstyle{remark}

\journal{\color{white} Insurance: Mathematics and Economics}

\begin{document}

\begin{frontmatter}



\title{Conditional Least Squares and Copulae in Claims Reserving for a~Single Line of Business}


\author[mff]{Michal Pe\v{s}ta\corref{cor1}} \ead{michal.pesta@mff.cuni.cz}

\author[hu]{Ostap Okhrin} \ead{ostap.okhrin@wiwi.hu-berlin.de}

\cortext[cor1]{Corresponding author, \emph{tel:} (+420) 221 913 400, \emph{fax:} (+420) 222 323 316}

\address[mff]{Charles University in Prague, Faculty of Mathematics and Physics, Department of Probability and Mathematical Statistics, Sokolovsk\'{a} 83, CZ-18675 Prague, Czech Republic}

\address[hu]{Humboldt University of Berlin, School of Business and Economics, Ladislaus von Bortkiewicz Chair of Statistics, C.A.S.E.~-- Center for Applied Statistics and Economics, Spandauer Strasse~1, D-10178 Berlin, Germany}

\begin{abstract}
One of the main goals in non-life insurance is to estimate the claims reserve distribution. A~generalized time series model, that allows for modeling the conditional mean and variance of the claim amounts, is proposed for the claims development. On contrary to the classical stochastic reserving techniques, the number of model parameters does not depend on the number of development periods, which leads to a~more precise forecasting.

Moreover, the time series innovations for the consecutive claims are not considered to be independent anymore. Conditional least squares are used for model parameter estimation and consistency of such estimate is proved. Copula approach is used for modeling the dependence structure, which improves the precision of the reserve distribution estimate as well.

Real data examples are provided as an illustration of the potential benefits of the presented approach.
\end{abstract}

\begin{keyword}
claims reserving \sep reserve distribution \sep dependency modeling \sep copula \sep conditional least squares


\begin{flushleft}\vspace{-.13cm}
\end{flushleft}

\begin{flushleft}\vspace{-.13cm}
\emph{Subject Category and Insurance Branch Category:} IM10, IM11, IM20, IM40
\end{flushleft}

\begin{flushleft}\vspace{-.13cm}
\end{flushleft}

\end{keyword}

\end{frontmatter}



\section{Introduction}
Claims reserving is one of the most important issues in general insurance. A~large number of various methods has been invented, see~\cite{england2002} or~\cite{wutrich_kniha} for an~overview.

Main aim of this paper is to deal with serious issues in contemporary reserving techniques, which are quite often set aside, but cause serious problems in the actuarial estimation and prediction. Such pitfalls are assumption of independent claims, independent stochastic errors (or residuals) in the corresponding claims reserving model, and considering large number of parameters often depending on the number of observations.

Majority of the classical approaches are based on the assumption that the claim amounts in different years are independent. However, this assumption can sometimes be unrealistic or at least questionable. It has been pointed out that methods, which enable~\emph{modeling the dependencies}, are needed, cf.~\cite{antonio2} or~\cite{HP2013}. The mentioned papers suggest the generalized linear mixed models (GLMM) or generalized estimating equations (GEE) to handle the possible dependence among the incremental claims in successive development years. These approaches extend the classical GLM and are frequently used in panel (longitudinal) data analyses. In this paper, we present another possible attitude, namely the \emph{conditional mean-variance} model with a~\emph{copula}.

On one hand, time series model by~\cite{BBMW2006} nicely and simply allow to model conditional mean and variance of the claim amounts. On the other hand, that model possesses two disadvantages, which are common for a~huge majority of the reserving methods: infinite number of parameters (i.e., depending on the number of observation) and independent errors. Generally, large number of parameters decreases the precision of estimation, because of not sufficient amount of data for estimation. Furthermore, the classical statistical inference is not valid anymore when the number of parameters depends on the number of observation. To overcome such difficulties, we consider a~generalized time series model with \emph{a~finite number of parameters not depending on the number of development periods} and, additionally, the \emph{model errors} belonging to the same accident period are \emph{not independent}.

Moreover, all the currently used bootstrap methods in claims reserving require independent residuals in order to estimate the distribution of the reserve and, consequently, calculate some distributional quantities, e.g., VaR at $99.5\%$. Assumption of independent residuals can be quite unrealistic in the claims reserving setup. Hence, an alternative and more suitable resampling method needs to be proposed in order to sensibly estimate the reserves distribution.

Copulae have already been utilized in the claims reserving to model dependences between different lines of business, e.g., \cite{SF2011}. On the contrary, it has to be emphasized that in our approach, only one line of business is taken into account. Copulae are therefore used to model dependences \emph{within claims} corresponding to that \emph{single line of business}. For sure, our approach can be generalized for several lines of business in the way that a~second level of dependence (for instance, modeled again by the copulae) is introduced between the claim amounts from different lines of business.

The structure of this paper is as follows: The claims reserving notation is summarized in Section~\ref{sec_claims}. In Section~\ref{sec_model}, a~generalized time series model for the conditional mean and variance of claim amounts is introduced. Section~\ref{sec_copula} elaborates copula approach for dependence modeling within the generalized time series model for claims triangles. Section~\ref{sec_est} covers estimation techniques for the parameters of the generalized time series model and copula as well. Consistency of the estimates is derived. Section~\ref{sec_prediction} concerns prediction of the actuarial claims reserves and, furthermore, estimation of their distribution. Finally, all the presented methods and approaches are applied on real data in Section~\ref{sec_data} in order to show their performance and outstanding benefits.

\section{Claims Reserving Notation}\label{sec_claims}
We introduce the classical claims reserving notation and terminology. Outstanding loss liabilities are structured in so-called claims development triangles, see Table~\ref{tab:run-off}. Let us denote $Y_{i,j}$ all the claim amounts up to development year $j\in\{1,\ldots,n\}$ with accident year $i\in\{1,\ldots,n\}$. Therefore, $Y_{i,j}$ stands for the \emph{cumulative claims} in accident year $i$ after $j$ development periods. The current year is $n$, which corresponds to the most recent accident year and development period as well. Hence, $Y_{i,j}$ is a~random variable of which we have an observation if $i+j<n+1$ (a~run-off triangle). That is, our data history consists of right-angled isosceles triangle $\{Y_{i,j}\}$, where $i=1,\ldots,n$ and $j=1,\ldots,n+1-i$. The diagonal elements $Y_{i,j}$, where $i+j$ is constant, correspond to the claim amounts in \emph{accounting year} $i+j$.
\begin{table}[!htb]
\renewcommand{\arraystretch}{1.1}
\begin{center}
\begin{tabular}{c|>{\centering}p{1.2cm}<{\centering}>{\centering}p{1.2cm}<{\centering}>{\centering}p{1.2cm}<{\centering}>{\centering}p{1.2cm}<{\centering}>{\centering}p{1.2cm}<{\centering}>{\centering}p{1.2cm}p{1.2cm}<{\centering}}
\hline
Accident & \multicolumn{7}{c}{Development year $j$} \\
\cline{2-8}
year $i$ & \cellcolor[gray]{.6} $1$ & \cellcolor[gray]{.6} $2$ & \cellcolor[gray]{.6} & \cellcolor[gray]{.6} $\cdots$ & \cellcolor[gray]{.6} & \cellcolor[gray]{.6} $n-1$ & \cellcolor[gray]{.6} $n$ \\
\hline
\cellcolor[gray]{.6} $1$ & \cellcolor[gray]{.8} $Y_{1,1}$ & \cellcolor[gray]{.8} $Y_{1,2}$ & \cellcolor[gray]{.8} & \cellcolor[gray]{.8} $\cdots$ & \cellcolor[gray]{.8} & \cellcolor[gray]{.8} $Y_{1,n-1}$ & \cellcolor[gray]{.8} $Y_{1,n}$ \\
\cellcolor[gray]{.6} $2$ & \cellcolor[gray]{.8} $Y_{2,1}$ & \cellcolor[gray]{.8} $Y_{2,2}$ & \cellcolor[gray]{.8} & \cellcolor[gray]{.8} $\cdots$ & \cellcolor[gray]{.8} & \cellcolor[gray]{.8} $Y_{2,n-1}$ & \\
\cellcolor[gray]{.6} & \cellcolor[gray]{.8} & \cellcolor[gray]{.8} & \cellcolor[gray]{.8} $\ddots$ & \cellcolor[gray]{.8} & \cellcolor[gray]{.8} && \\
\cellcolor[gray]{.6} $\vdots$ & \cellcolor[gray]{.8} $\vdots$ & \cellcolor[gray]{.8} $\vdots$ & \cellcolor[gray]{.8} & \cellcolor[gray]{.8} $Y_{i,n+1-i}$ &&& \\
\cellcolor[gray]{.6} & \cellcolor[gray]{.8} & \cellcolor[gray]{.8} & \cellcolor[gray]{.8} &&&& \\
\cellcolor[gray]{.6} $n-1$ & \cellcolor[gray]{.8} $Y_{n-1,1}$ & \cellcolor[gray]{.8} $Y_{n-1,2}$ &&&&& \\
\cellcolor[gray]{.6} $n$ & \cellcolor[gray]{.8} $Y_{n,1}$ &&&&&& \\
\hline
\end{tabular}
\end{center}
\caption{Run-off triangle for cumulative claim amounts $Y_{i,j}$.}
\label{tab:run-off}
\renewcommand{\arraystretch}{1.0}
\end{table}

The aim is to estimate the ultimate claims amount $Y_{i,n}$ and the outstanding \emph{claims reserve} $R_i^{(n)}=Y_{i,n}-Y_{i,n+1-i}$ for all $i=2,\dots,n$. Additional to that, it is needed to \emph{estimate the whole distribution of the reserves} in order to provide important distributional quantities for the \emph{Solvency II purposes}, e.g., quantiles for the value at risk calculation.

\section{Conditional Mean and Variance Model}\label{sec_model}
Run-off triangles are comprised by observations which are ordered in time. It is therefore natural to suspect the observations to be dependent. On one hand, the most natural approach is to assume that the observations of a~common accident year are dependent. On the other hand, observations of different accident years are supposed to be independent. This assumption is similar to those of the Mack's chain ladder model, cf.~\cite{mack}.

$\mathcal{F}_{i,j}$ denotes the information set generated by trapezoid $\{Y_{k,l}:l\leq j,k\leq i+1-j\}$, i.e., $\mathcal{F}_{i,j}=\sigma(Y_{k,l}:l\leq j,k\leq i+1-j)$ is a~filtration corresponding to the smallest $\sigma$-algebra containing historical claims with at most $j$ development periods paid in accounting period $i$ or earlier. This notation allows for zero or even negative index in filtration despite the fact that the claims corresponding to zero or negative development of accident years are not observed.

Let us define a~\emph{nonlinear generalized semiparametric regression} type of model. It can be considered as a~generalization of the model proposed by~\cite{BBMW2006}. The first level of generalization is in the mean and variance structure, which was inspired by~\cite{Patton2012}. The second level of generalization regarding the dependence structure will be introduced in the next Section~\ref{sec_copula}.

\begin{definition}[CMV~model]\label{def:CMVmodel}
The Conditional Mean and Variance (CMV) model assumes
\begin{equation}\label{CMV}
Y_{i,j}=\mu(Y_{i,j-1},\balpha,j)+\sigma(Y_{i,j-1},\bbeta,j)\eps_{i,j}(\balpha,\bbeta),
\end{equation}
where $\balpha$ and $\bbeta$ are unknown parameters, which dimensions do not depend on $n$, $\mu$ is a~continuous function in~$\balpha$ and $\sigma$ is a~positive and continuous function in~$\bbeta$. Disturbances $\{\eps_{i,j}(\balpha,\bbeta)\}_{j=1}^{n+1-i}$ are independent sample copies of a~stationary first-order Markov process for all $i$. All $\eps_{i,j}(\balpha,\bbeta)$ have the common true invariant distribution $G_{\balpha,\bbeta}$ which is absolutely continuous with respect to Lebesgue measure on the real line. Suppose that
\begin{subequations}\label{eps-ass}
\begin{align}
\E[\eps_{i,j}(\balpha,\bbeta)|\mathcal{F}_{i,j-1}]&=0,\\
\var[\eps_{i,j}(\balpha,\bbeta)|\mathcal{F}_{i,j-1}]&=s(\balpha,\bbeta),
\end{align}
\end{subequations}
for all $i$ and $j$. Moreover for the unknown true values $[\balpha^{*\top},\bbeta^{*\top}]^{\top}$ of parameters $[\balpha^{\top},\bbeta^{\top}]^{\top}$, the conditional variance of errors equals one due to identifiability purposes, i.e., $s(\balpha^*,\bbeta^*)=1$.
\end{definition}

The name of the model comes from the fact that the conditional mean and variance can be expressed as
\begin{align*}
\E[Y_{i,j}|\mathcal{F}_{i,j-1}]&=\mu(Y_{i,j-1},\balpha,j),\\
\var[Y_{i,j}|\mathcal{F}_{i,j-1}]&=\sigma^2(Y_{i,j-1},\bbeta,j)s(\balpha,\bbeta).
\end{align*}
This property allows for a~wide variety of models for the conditional mean: types of ARMA models, vector autoregressions, linear and nonlinear regressions, and others. It also allows for a variety of models for the conditional variance: ARCH and any of its numerous parametric extensions (GARCH, EGARCH, GJR-GARCH, etc., see \cite{Bollerslev2010}), stochastic volatility models, and others.

\cite{Patton2012} considered a~similar model, but the dependence was assumed in a~different way. I.e., dependent copies of the time series (dependence between rows) were supposed, not dependent errors within each time series as we propose. Here, independent rows of errors $[\eps_{i,1}(\balpha,\bbeta),\ldots,\eps_{i,n+1-i}(\balpha,\bbeta)]$ imply independent rows of claims $[Y_{i,1},\ldots,Y_{i,n+1-i}]$. Moreover, the unconditional mean and variance of the CMV model's errors equal the conditional ones: $\E\eps_{i,j}(\balpha,\bbeta)=0$ and $\var\eps_{i,j}(\balpha,\bbeta)=s(\balpha,\bbeta)$.

To provide an~insight into possible candidates for the mean function $\mu$ and variance function $\sigma$, one may propose $\mu(Y_{i,j-1},\balpha,j)=\eta(\balpha,j)Y_{i,j-1}$ and $\sigma(Y_{i,j-1},\bbeta,j)=\nu(\bbeta,j)\sqrt{Y_{i,j-1}}$ or $\sigma(Y_{i,j-1},\bbeta,j)=\nu(\bbeta,j)Y_{i,j-1}$. \cite{Sherman1984} investigated decays $\eta(\balpha,j)$, which should correspond to the link ratios. Hence, $\eta(\balpha,j)$ should be decreasing in $j$ with limit $1$ as $j$ tends to infinity: $1+\alpha_1\exp\{-\alpha_2 j\}$, $1+\alpha_1\exp\{-j^{\alpha_2}\}$, $1+\alpha_1 j^{-\alpha_2}$, $1+\alpha_1(\alpha_2+j)^{-\alpha_3}$, $1+\alpha_1^{-\alpha_2^j}$, $\alpha_1^{\alpha_2^{-j}}$, $(1-\exp\{-\alpha_1 j^{\alpha_2}\})^{-1}$, $1+\alpha_1\alpha_2j^{-1-\alpha_2}\exp\{\alpha_1j^{-\alpha_2}\}$, $\exp\{\alpha_1j^{-\alpha_2}\}$, $1+\alpha_1/(j+\alpha_2)$, $1+\alpha_1/\log(j+\alpha_2)$, etc., where $\alpha_1,\alpha_2,\alpha_3>0$. On the other hand, decay $\nu(\bbeta,j)$ should be decreasing in $j$ with limit $0$: $\beta_1\exp\{-\beta_2j\}$, $\beta_1j^{-\beta_2}$, $\beta_1\log j\exp\{-\beta_2j\}$, $\beta_1j^{-\beta_2}\exp\{-\beta_3j\}$, $\beta_1\exp\{-\beta_2j^2\}$, $\beta_1/(j+\beta_2)$, $\beta_1/\log(j+\beta_2)$, $\beta_1j^{-\beta_2}\exp\{-\beta_3j^2\}$, etc., where $\beta_1,\beta_2,\beta_3>0$.

In actuarial praxis, these decays are quite often used, mainly for \emph{projecting the development} (forecasting of the claim amounts after $n$ development periods). Despite of that, the parameters of the decay curves are not estimated directly from the triangle, but the chain ladder estimates $\widehat{f}_j$ and $\widehat{\sigma}^2_j$ of the development factors $f_j$ and the nuisance variance parameters $\sigma_j^2$ \citep{mack} are smoothed and used for the decay parameters estimation. This two-step procedure does not assure that the estimated decay parameters will be at least asymptotically unbiased. Unlike that, we will estimate the parameters \emph{directly from the data triangle} and prove the estimates' consistency.

\cite{Sherman1984} assumed independence of individual link ratios (development factors) when estimating the decay parameters by parametric curve fitting. We relax the independence assumption and model the link ratios conditionally having dependent errors.

When comparing the CMV model with the model investigated by~\cite{BBMW2006}, two main differences arise. The CMV model allows for dependent errors and assumes finite number of parameters not depending on number of development periods $n$. Indeed, the CMV model requires known functions with unknown finite dimensional parameters. Parameters of the time series model from~\cite{BBMW2006} are $\{f_j\}_{j=1}^n$ and $\{\sigma_j\}_{j=1}^n$, which play the role of $\eta(\balpha,j)$ and $\nu(\bbeta,j)$, respectively. It is important to note that the classical stochastic inference is not valid in the setup, when number of parameters depends on the number of observation. Thus, legitimacy of the bootstrap procedure in that case is questionable.

Furthermore in the chain ladder, the estimate for $f_{n-1}$ is just a~pure ratio of two random variables and, moreover to estimate $\sigma^2_{n-1}$, only doubtful ad-hoc estimates were proposed due to the fact that the claims triangle simply does not contain data for a~reasonable estimate (e.g., a~consistent one).



\section{Dependence Modeling by Copulae}\label{sec_copula}
Since the mean and variance trend are removed by the CMV model, the rest of the relationship among claim amounts can be additionally captured by modeling dependent errors. The inspiration for the dependence structure was taken from~\cite{ChenFan2006a}.

\begin{paragraph}{Assumption~C}
$\{\eps_{i,j}(\balpha,\bbeta)\}_{j=1}^{n+1-i}$ are independent sample copies of a~stationary first-order Markov process for all $i$ generated from $(G_{\balpha,\bbeta}(\cdot),C(\cdot,\cdot;\bgamma))$, where $C(\cdot,\cdot;\bgamma)$ is the true parametric copula for $[\eps_{i,j-1}(\balpha,\bbeta),\eps_{i,j}(\balpha,\bbeta)]$, which is given and fixed up to unknown parameter $\bgamma$ and is absolutely continuous with respect to Lebesgue measure on $[0,1]^2$.
\end{paragraph}


It is believed that there exist a~kind of \emph{information overlap} between the claims from consecutive development periods, which corresponds to the dependence between the CMV model's errors modeled by copulae.

Assumption~C together with the CMV model yield a~copula-based model, where the joint bivariate distribution of errors $[\eps_{i,j-1}(\balpha,\bbeta),\eps_{i,j}(\balpha,\bbeta)]$ has the following distribution function
\[
H(e_1,e_2)=C(G_{\balpha,\bbeta}(e_1),G_{\balpha,\bbeta}(e_2);\bgamma).
\]
Then, the conditional copula density can be derived as
\begin{equation}\label{copula-cond}
h(e_2|e_1)=g_{\balpha,\bbeta}(e_2)c(G_{\balpha,\bbeta}(e_1),G_{\balpha,\bbeta}(e_2);\bgamma),
\end{equation}
where $c$ is the copula density a $g_{\balpha,\bbeta}$ is the marginal density corresponding to the univariate distribution function $G_{\balpha,\bbeta}$. The latter relation~\eqref{copula-cond} will play an~important role in ``making" the dependent errors \emph{conditionally independent} during the forthcoming estimation and prediction process.

\section{Parameter Estimation}\label{sec_est}
The CMV model from Definition~\ref{def:CMVmodel} together with the copula Assumption~C contain three vector parameters, which need to be estimated. The estimation process consists of two stages. In the first one, mean and variance parameters $\balpha$ and $\bbeta$ are estimated in a~distribution-free fashion, since no specific distributional assumptions are proposed nor required for the claims. The second stage concerns estimation of the dependence structure, mainly the copula parameter $\bgamma$, in a~likelihood based way.

\subsection{Estimation in CMV Model}
Since the CMV model is defined in a~conditional style, \emph{conditional least squares} (CLS) of the sample centered conditional moments of the claims are minimized in order to obtain estimates of the CMV model parameters.

\begin{definition}[Conditional least squares estimates]\label{def:CLSest}
Let us denote
\[
M_n(\balpha,\bbeta)=\frac{1}{n-1}\sum_{j=2}^{n}\frac{1}{n+1-j}\sum_{i=1}^{n+1-j}\frac{\left[Y_{i,j}-\mu(Y_{i,j-1},\balpha,j)\right]^2}{\sigma^2(Y_{i,j-1},\bbeta,j)}
\]
and
\begin{multline*}
V_n(\balpha,\bbeta)\\=\frac{1}{n-1}\sum_{j=2}^{n}\frac{1}{n+1-j}\sum_{i=1}^{n+1-j}\left\{\left[Y_{i,j}-\mu(Y_{i,j-1},\balpha,j)\right]^2-\sigma^2(Y_{i,j-1},\bbeta,j)\right\}^2,
\end{multline*}
where parameters $\balpha$ and $\bbeta$ belong to parameter spaces $\btheta_1$ and $\btheta_2$. The conditional least squares estimate of the mean parameter $\balpha$ for a~fixed value of parameter $\bbeta\in\btheta_2$ is defined as
\[
\widehat{\balpha}(\bbeta)=\arg\min_{\balpha\in\btheta_1}M_n(\balpha,\bbeta)
\]
and the conditional least squares estimate of the variance parameter $\bbeta$ for a~fixed value of parameter $\balpha\in\btheta_1$ is defined as
\[
\widehat{\bbeta}(\balpha)=\arg\min_{\bbeta\in\btheta_2}V_n(\balpha,\bbeta).
\]
\end{definition}

The reason, why the parameter estimates for the CMV model are defined as above, lies in the fact that it is computationally not feasible to find the global minimum of $M_n$ and $V_n$ with respect to $[\balpha^{\top},\bbeta^{\top}]^{\top}$ simultaneously.

The forthcoming theory (Theorems~\ref{thm:CLSconsistencyM}, \ref{thm:CLSconsistencyV}, and Corollary~\ref{col:consistency}) assures that the CLS estimates are reasonable and, moreover, consequent Algorithm~\ref{alg:icls} provides a~computational way for obtaining CMV parameter estimates.

\begin{theorem}[Conditional least squares consistency for the mean]\label{thm:CLSconsistencyM}
Let Model CMV hold and $\bbeta\in\btheta_2$ be fixed. Assume that
\begin{enumerate}
\item $\left\{\left[Y_{i,j}-\mu(Y_{i,j-1},\balpha,j)\right]^2/\sigma^2(Y_{i,j-1},\bbeta,j)\right\}_{i,j\in\mathbbm{N}}$ is uniformly integrable,
\item for all $\balpha,\balpha'\in\btheta_1$ and $n\in\mathbbm{N}$,
\begin{equation}\label{cond:lipschitz1}
\left|\left[Y_{i,j}-\mu(Y_{i,j-1},\balpha,j)\right]^2-\left[Y_{i,j}-\mu(Y_{i,j-1},\balpha',j)\right]^2\right|
\stackrel{a.s.}{\leq} C_j g(\|\balpha-\balpha'\|),
\end{equation}
where $\{C_j\}_{j\in\mathbbm{N}}$ is a~stochastic sequence not depending on~$\balpha$ such that $C_j=\mathcal{O}_{\P}(\sigma^2(Y_{i,j-1},\bbeta,j))$, $j\to\infty$ for all $i\in\mathbbm{N}$ and $g$ is nonstochastic such that $g(t)\downarrow 0$ as $t\downarrow 0$,
\item $s(\cdot,\bbeta)$ is a~Lipschitz function on the compact parameter space $\btheta_1$ such that the true unknown parameter $\balpha^*(\bbeta)$ is its unique global minimum.
\end{enumerate}
Then $\widehat{\balpha}(\bbeta)\xrightarrow[n\to\infty]{\P}\balpha^*(\bbeta)$.
\end{theorem}

\begin{proof}
Let us define
\[
A_{n,j}:=\frac{1}{(n-1)(n+1-j)}\sum_{i=1}^{n+1-j}\frac{\left[Y_{i,j}-\mu(Y_{i,j-1},\balpha,j)\right]^2}{\sigma^2(Y_{i,j-1},\bbeta,j)}.
\]
Since $\E[A_{n,j}|\mathcal{F}_{n,j-1}]=\frac{1}{n-1}s(\balpha,\bbeta)$, then $A_{n,j}-\E[A_{n,j}|\mathcal{F}_{n,j-1}]=A_{n,j}-\E A_{n,j}$ is a~martingale difference array with respect to filtration $\mathcal{F}_{n,j}$. Moreover, if $\left[Y_{i,j}-\mu(Y_{i,j-1},\balpha,j)\right]^2/\sigma^2(Y_{i,j-1},\bbeta,j)$ is uniformly integrable, then $(n-1)A_{n,j}$ is uniformly integrable as well. This allows to apply the $\mathsf{L}_1$~law of large numbers for the martingale difference arrays \citep[Theorem~19.7]{Davidson1994}. Hence, $\sum_{j=2}^{n}[A_{n,j}-\frac{1}{n-1}s(\balpha,\bbeta)]\xrightarrow[n\to\infty]{\mathsf{L}_1}0$, which implies $\sum_{j=2}^{n}A_{n,j}\xrightarrow[n\to\infty]{\P}s(\balpha,\bbeta)$. According to the definition of $A_{n,j}$, we have obtained the weak law of large numbers (WLLN) for $M_n(\balpha,\bbeta)$ for all $\balpha\in\btheta_1$ and $\bbeta\in\btheta_2$ (pointwise).

Assumption~\eqref{cond:lipschitz1} gives a~Lipschitz-type condition for $M_n(\cdot,\bbeta)$, i.e.,
\[
\left|M_n(\balpha,\bbeta)-M_n(\balpha',\bbeta)\right|\leq K_n g(\|\balpha-\balpha'\|),\quad a.s.
\]
for all $\balpha,\balpha'$ and $K_n=\mathcal{O}_{\P}(1),\,n\to\infty$. Combining this fact with $s(\cdot,\bbeta)$ being Lipschitz, then $\{M_n(\cdot,\bbeta)-s(\cdot,\bbeta)\}_{n\in\mathbbm{N}}$ is stochastically equicontinuous for all $\bbeta$ by \citet[Theorem~21.10]{Davidson1994}. Furthermore, Theorem~21.9 by \cite{Davidson1994} provides the weak uniform law of large numbers (WULLN) for $M_n(\balpha,\bbeta)$ in $\balpha$ for all $\bbeta\in\btheta_2$, i.e.,
\[
\sup_{\balpha\in\btheta_1}|M_n(\balpha,\bbeta)-s(\balpha,\bbeta)|\xrightarrow[n\to\infty]{\P}0
\]
for all $\bbeta\in\btheta_2$.

Taking into account that continuous (or, moreover, Lipschitz) functions reach their global extremes on a~compact set, Theorem~4.2.1 by \cite{Bierens1994} yields the desired weak consistency of the $\balpha$ parameter.
\end{proof}

Weak consistency (in probability) of mean parameter estimate is shown, but also the strong version (almost sure convergence) can be provided. It would require $C_j$ to be bounded almost surely, which is less feasible.

Lipschitz kind of assumption~(ii) can be replaced by a~stronger one: uniform equiboundedness in probability. In that case, it suffices to assume $\E\sup_{\balpha\in\btheta_1}\|\nabla_{\balpha}M_n(\balpha,\bbeta)\|<\Delta_1$ for all $n$ and $\bbeta$ and convexity of the compact parameter space $\btheta_1$ for applying the stochastic mean-value theorem. The compactness of the parameter space can even be relaxed to its total boundedness.


Similar theorem as above is going to be postulated for the CLS variance parameter estimate to ensure its appropriateness. Firstly, let us define
\begin{multline*}
B_{n,j}(\balpha,\bbeta)\\:=\frac{1}{(n-1)(n+1-j)}\sum_{i=1}^{n+1-j}\left\{\left[Y_{i,j}-\mu(Y_{i,j-1},\balpha,j)\right]^2-\sigma^2(Y_{i,j-1},\bbeta,j)\right\}^2
\end{multline*}
and
\[
v(\balpha,\bbeta):=\lim_{n\to\infty}\sum_{j=2}^{n}\E B_{n,j}(\balpha,\bbeta).
\]

\begin{theorem}[Conditional least squares consistency for the variance]\label{thm:CLSconsistencyV}
Let Model CMV hold and $\balpha\in\btheta_1$ be fixed. Assume that
\begin{enumerate}
\item random array
\begin{multline*}
\left\{\left\{\left[Y_{i,j}-\mu(Y_{i,j-1},\balpha,j)\right]^2-\sigma^2(Y_{i,j-1},\bbeta,j)\right\}^2\right.\\
\left.-\E\left\{\left[Y_{i,j}-\mu(Y_{i,j-1},\balpha,j)\right]^2-\sigma^2(Y_{i,j-1},\bbeta,j)\right\}^2\right\}_{i,j\in\mathbbm{N}}
\end{multline*}
is uniformly integrable,
\item for all $\bbeta,\bbeta'\in\btheta_2$ and $n\in\mathbbm{N}$,
\begin{multline}\label{cond:lipschitz2}
\left|\left\{\left[Y_{i,j}-\mu(Y_{i,j-1},\balpha,j)\right]^2-\sigma^2(Y_{i,j-1},\bbeta,j)\right\}^2\right.\\
\left.-\left\{\left[Y_{i,j}-\mu(Y_{i,j-1},\balpha,j)\right]^2-\sigma^2(Y_{i,j-1},\bbeta',j)\right\}^2\right|\stackrel{a.s.}{\leq} D_j h(\|\bbeta-\bbeta'\|),
\end{multline}
where $\{D_j\}_{j\in\mathbbm{N}}$ is a~stochastic sequence not depending on~$\bbeta$ such that $D_j=\mathcal{O}_{\P}(1),\,j\to\infty$ for all $i\in\mathbbm{N}$ and $h$ is nonstochastic such that $h(t)\downarrow 0$ as $t\downarrow 0$,
\item for all $j\leq n$, $n\in\mathbbm{N}$, and $m\in\mathbbm{N}_0$,
\begin{equation}\label{cond:mixingale}
\E\left|\E\left[B_{n,j}(\balpha,\bbeta)-\E B_{n,j}(\balpha,\bbeta)|\mathcal{F}_{n,j-m}\right]\right|\leq c_{n,j}d_m,
\end{equation}
where $\{c_{n,j}\}_{j\leq n,n\in\mathbbm{N}}$ and $\{d_m\}_{m\in\mathbbm{N}_0}$ are constants such that $d_m\downarrow 0$ as $m\to\infty$,
\item $v(\balpha,\cdot)$ is a~Lipschitz function on the compact parameter space $\btheta_2$ such that the true unknown parameter $\bbeta^*(\balpha)$ is its unique global minimum.
\end{enumerate}
Then $\widehat{\bbeta}(\balpha)\xrightarrow[n\to\infty]{\P}\bbeta^*(\balpha)$.
\end{theorem}

\begin{proof}
The idea of this proof is similar as in the previous one. Note that $B_{n,j}(\balpha,\bbeta)-\E B_{n,j}(\balpha,\bbeta)$ is an~$\mathsf{L}_1$-mixingale array with respect to filtration $\mathcal{F}_{n,j}$ due to~\eqref{cond:mixingale}. Condition~(i) implies that $(n-1)[B_{n,j}(\balpha,\bbeta)-\E B_{n,j}(\balpha,\bbeta)]$ is uniformly integrable. This allows us to apply the $\mathsf{L}_1$~law of large numbers for the $\mathsf{L}_1$-mixingale arrays \citep[Theorem~19.11]{Davidson1994}. Hence, $\sum_{j=2}^{n}[B_{n,j}(\balpha,\bbeta)-\E B_{n,j}(\balpha,\bbeta)]\xrightarrow[n\to\infty]{\mathsf{L}_1}0$, which implies $\sum_{j=2}^{n}B_{n,j}(\balpha,\bbeta)\xrightarrow[n\to\infty]{\P}v(\balpha,\bbeta)$. According to the definition of $B_{n,j}(\balpha,\bbeta)$, we have obtained the pointwise WLLN for $V_n(\balpha,\bbeta)$ for all $\balpha\in\btheta_1$ and $\bbeta\in\btheta_2$.

Assumption~\eqref{cond:lipschitz2} gives a~Lipschitz-type condition for $V_n(\balpha,\cdot)$, i.e.,
\[
\left|V_n(\balpha,\bbeta)-V_n(\balpha,\bbeta')\right|\leq D'_n h(\|\bbeta-\bbeta'\|),\quad a.s.
\]
for all $\bbeta,\bbeta'$ and $D'_n=\mathcal{O}_{\P}(1),\,n\to\infty$. Combining this fact with $v(\balpha,\cdot)$ being Lipschitz, then $\{V_n(\balpha,\cdot)-v(\balpha,\cdot)\}_{n\in\mathbbm{N}}$ is stochastically equicontinuous for all $\balpha$ by \citet[Theorem~21.10]{Davidson1994}. Furthermore, Theorem~21.9 by \cite{Davidson1994} provides the WULLN for $V_n(\balpha,\bbeta)$ in $\bbeta$ for all $\balpha\in\btheta_1$, i.e.,
\[
\sup_{\bbeta\in\btheta_2}|V_n(\balpha,\bbeta)-v(\balpha,\bbeta)|\xrightarrow[n\to\infty]{\P}0
\]
for all $\balpha\in\btheta_1$. Theorem~4.2.1 by \cite{Bierens1994} yields the desired weak consistency of the $\bbeta$ parameter.
\end{proof}

A~natural question arises: What is the connection between the true unknown parameter values $\balpha^*$ and $\bbeta^*$ of the CMV model and true unknown parameter values $\balpha^*(\bbeta)$ and $\bbeta^*(\balpha)$ from Theorems~\ref{thm:CLSconsistencyM} and~\ref{thm:CLSconsistencyV}? The intuition behind the CMV model is that function $\mu$ should mimic the conditional mean of the claims and function $\sigma^2$ should model their conditional variance. Mathematically speaking, $\var[Y_{i,j}/\sigma(Y_{i,j-1},\bbeta,j)|\mathcal{F}_{i,j-1}]$ and, similarly, $\var[(Y_{i,j}-\mu(Y_{i,j-1},\balpha,j))^2|\mathcal{F}_{i,j-1}]$ should be as small as possible. Taking into account that the data triangle does not possess the same number of claim amounts entries for each development period $j$, it is reasonable to assume that if the CMV model holds, then both discrepancy measures
\begin{equation}\label{dm1}
\lim_{n\to\infty}\E\left\{\frac{1}{n-1}\sum_{j=2}^{n}\frac{1}{n+1-j}\sum_{i=1}^{n+1-j}\var\left[\frac{Y_{i,j}}{\sigma(Y_{i,j-1},\bbeta,j)}\Big|\mathcal{F}_{i,j-1}\right]\right\}
\end{equation}
and
\begin{equation}\label{dm2}
\lim_{n\to\infty}\E\left\{\frac{1}{n-1}\sum_{j=2}^{n}\frac{1}{n+1-j}\sum_{i=1}^{n+1-j}\var\left[\left(Y_{i,j}-\mu(Y_{i,j-1},\balpha,j)\right)^2|\mathcal{F}_{i,j-1}\right]\right\}
\end{equation}
reach their global minimum just at the same true unknown parameter values $\balpha^*$ and $\bbeta^*$ of the CMV model. However, measures~\eqref{dm1} and \eqref{dm2} are nothing else than $s(\balpha,\bbeta)$ and $v(\balpha,\bbeta)$. Now, let us define the interior of set~$\btheta$ by $\interior{\btheta}$.

\begin{corollary}[Consistency of the CLS estimates]\label{col:consistency}
Suppose that the assumptions of Theorems~\ref{thm:CLSconsistencyM} and~\ref{thm:CLSconsistencyV} hold. Let $s\in\mathcal{C}^2(\btheta_1\times\btheta_2)$, $v\in\mathcal{C}^2(\btheta_1\times\btheta_2)$, and both functions $s$ and $v$ have their unique global minimum on compact set $\btheta_1\times\btheta_2$ at $[\balpha^{*\top},\bbeta^{*\top}]^{\top}\in\interior{\btheta}_1\times\interior{\btheta}_2$. If $\det[\partial^2s(\balpha,\bbeta)/\partial\balpha\partial\balpha^{\top}]\neq 0$ and $\det[\partial^2v(\balpha,\bbeta)/\partial\bbeta\partial\bbeta^{\top}]\neq 0$ for all $[\balpha^{\top},\bbeta^{\top}]^{\top}\in\interior{\btheta}_1\times\interior{\btheta}_2$, then
\[
\left[\begin{array}{c}\widehat{\balpha}(\bbeta^*)\\\widehat{\bbeta}(\balpha^*)\end{array}\right]\xrightarrow[n\to\infty]{\P}\left[\begin{array}{c}\balpha^*\\\bbeta^*\end{array}\right].
\]

\end{corollary}

\begin{proof}
Let us define $\varphi(\balpha,\bbeta):=\partial s(\balpha,\bbeta)/\partial\balpha$ and $\varrho(\balpha,\bbeta):=\partial v(\balpha,\bbeta)/\partial\bbeta$. Since $\varphi(\balpha^*(\bbeta),\bbeta)=\zero$ for all $\bbeta\in\interior{\btheta}_2$, then for all $\bbeta\in\interior{\btheta}_2$ by the implicit function theorem, there exists a~unique function $\phi\in\mathcal{C}^1$ on the open surrounding of $\bbeta$ such that $\phi(\bbeta)=\balpha$. Due to this uniqueness, $\phi(\cdot)\equiv\balpha^*(\cdot)$.

Similarly, for all $\balpha\in\interior{\btheta}_1$, $\varrho(\balpha,\bbeta^*(\balpha))=\zero$ and, thus, there exists a~unique function $\rho\in\mathcal{C}^1$ on the open surrounding of $\balpha$ such that $\rho(\balpha)=\bbeta$. Hence, $\rho(\cdot)\equiv\bbeta^*(\cdot)\in\mathcal{C}^1(\interior{\btheta}_1)$.

Consequently, since $\varphi(\balpha^*(\bbeta^*),\bbeta^*)=0$ and $\varrho(\balpha^*,\bbeta^*(\balpha^*))=0$, then $\balpha^*(\bbeta^*)=\bbeta^*$ and $\bbeta^*(\balpha^*)=\balpha^*$. Nevertheless by Theorems~\ref{thm:CLSconsistencyM} and~\ref{thm:CLSconsistencyV},
\[
\left[\begin{array}{c}\widehat{\balpha}(\bbeta^*)\\\widehat{\bbeta}(\balpha^*)\end{array}\right]\xrightarrow[n\to\infty]{\P}\left[\begin{array}{c}\balpha^*(\bbeta^*)\\\bbeta^*(\balpha^*)\end{array}\right]=\left[\begin{array}{c}\balpha^*\\\bbeta^*\end{array}\right].
\]
\end{proof}

Variance parameter $\bbeta$ can be viewed as a~nuisance parameter when estimating the mean parameter $\balpha$ and vice-versa. The idea of joint estimation of $[\balpha^{\top},\bbeta^{\top}]^{\top}$ is to alternately perform partial optimizations from Definition~\ref{def:CLSest}. In fact, we iteratively estimate $\balpha$ given the fixed value of $\bbeta$ and, consequently, we estimate $\bbeta$ given the fixed value of $\balpha$ (obtained from previous step). This two steps are repeated in turns until almost no change in consecutive estimates of $[\balpha^{\top},\bbeta^{\top}]^{\top}$, see Algorithm~\ref{alg:icls}. Based on Corollary~\ref{col:consistency}, it is believed that each turn will bring our iterated estimates closer to the true unknown parameter values. Moreover, Algorithm~\ref{alg:icls} can be modified: the initial value of $\balpha^{(0)}$ could be required on the input instead of $\bbeta^{(0)}$ and the whole iteration procedure would start with the estimation of $\bbeta^{(1)}$ instead of $\balpha^{(1)}$.

\begin{algorithm}[!htb]
\onehalfspacing\caption{Iterative conditional least squares estimation of $\balpha$ and $\bbeta$.}
\label{alg:icls}
\algsetup{indent=2em}
\begin{algorithmic}[1]
\REQUIRE Cumulative claims triangle $\{Y_{i,j}\}_{i,j=1}^{n,n+1-i}$, mean and variance functions $\mu$ and $\sigma$, initial (starting) parameter value $\bbeta^{(0)}$, maximum number of iterations $M$, and convergence precision $\epsilon$.
\ENSURE CLS parameter estimates $\widehat{\balpha}$ and $\widehat{\bbeta}$, fitted residuals $\{\widehat{\eps}_{i,j}\}_{i=1,j=2}^{n-1,n+1-i}$.
\STATE $m \leftarrow 1$ and $\balpha^{(0)} \leftarrow\zero$
\STATE $\balpha^{(1)} \leftarrow \arg\min_{\balpha\in\btheta_1}M_n(\balpha,\bbeta^{(0)})$
\STATE $\bbeta^{(1)} \leftarrow \arg\min_{\bbeta\in\btheta_2}V_n(\balpha^{(1)},\bbeta)$
\WHILE{$m\leq M$ \AND $\|[\balpha^{(m)\top},\bbeta^{(m)\top}]^{\top}-[\balpha^{(m-1)\top},\bbeta^{(m-1)\top}]^{\top}\|>\epsilon$}
\STATE $\balpha^{(m+1)} \leftarrow \arg\min_{\balpha\in\btheta_1}M_n(\balpha,\bbeta^{(m)})$
\STATE $\bbeta^{(m+1)} \leftarrow \arg\min_{\bbeta\in\btheta_2}V_n(\balpha^{(m+1)},\bbeta)$
\STATE $m \leftarrow m+1$
\ENDWHILE
\STATE $\widehat{\balpha} \leftarrow \balpha^{(m)}$ and $\widehat{\bbeta} \leftarrow \bbeta^{(m)}$
\FOR{$i=1$ \TO $n-1$}
\FOR{$j=2$ \TO $n+1-i$}
\STATE $\widehat{\eps}_{i,j} \leftarrow [Y_{i,j}-\mu(Y_{i,j-1},\widehat{\balpha},j)]/\sigma(Y_{i,j-1}\widehat{\bbeta},j)$
\ENDFOR
\ENDFOR
\end{algorithmic}
\end{algorithm}

In order to demonstrate the \emph{finite-sample performance} of the CLS estimates, a~\emph{small simulation study} is provided. 200 triangles with $n=11$ accident years are simulated. The claim amounts in the first column of each triangle are $iid$ gamma distributed with mean and also variance equal to $10^5$. The mean and variance functions of the CMV model are chosen as in the forthcoming real data example (Section~\ref{sec_data}), see CMV structure~\eqref{CMVstructure}. The marginal distribution of the errors is the standard gaussian and the bivariate copula is the Gumbel copula with parameter $2$. The true unknown parameter values together with the sample means and sample standard deviations of their CLS estimates based on 200 triangle samples are shown in Table~\ref{tab:simul}. All four sample means are close to the corresponding true values, which is in concordance with the asymptotic theory derived above.

\begin{table}[!htb]
\renewcommand{\arraystretch}{1.1}
\begin{center}
\begin{tabular}{lrrrr}
\toprule
Parameter & $\alpha_1$ & $\alpha_2$ & $\beta_1$ & $\beta_2$ \tabularnewline
\midrule
True value & $2.000$ & $1.000$ & $100.000$ & $0.500$\tabularnewline
Mean & $2.001$ & $1.001$ & $101.214$ & $0.508$\tabularnewline
Standard deviation & $0.029$ & $0.031$ & $46.405$ & $0.131$\tabularnewline
\bottomrule
\end{tabular}
\end{center}
\caption{Sample means and sample standard deviations of 200 CLS estimates for the true parameter values (based on simulations).}
\label{tab:simul}
\renewcommand{\arraystretch}{1.0}
\end{table}




\subsection{Estimation of Dependence Structure}
The second stage of the parameter estimation process involves the estimation of the whole dependence structure in the claims triangle. Indeed, the strict stationarity of the first order Markov process imposed on the CMV model errors (Assumption~C) together with Sklar's theorem arrange that only the copula parameter and the marginal distribution of the errors are necessary to know the bivariate distribution of two in row neighbouring errors.

Since the estimates of the CMV model parameters are already available, one can estimate the unknown \emph{marginal distribution} function $G_{\balpha,\bbeta}$ of CMV model errors $\eps_{i,j}(\balpha,\bbeta)$ non-parametrically by the empirical distribution function
\[
\widehat{G}_n(e)=\frac{1}{n(n-1)/2+1}\sum_{	i=1}^{n-1}\sum_{j=2}^{n+1-i}\mathcal{I}\{\widehat{\eps}_{i,j}(\widehat{\balpha},\widehat{\bbeta})\leq e\},
\]
of the fitted residuals
\[
\widehat{\eps}_{i,j}(\widehat{\balpha},\widehat{\bbeta})=\frac{Y_{i,j}-\mu(Y_{i,j-1},\widehat{\balpha},j)}{\sigma(Y_{i,j-1},\widehat{\bbeta},j)}.
\]
The consistency of the CMV model parameter estimates $\widehat{\balpha}$ and $\widehat{\bbeta}$ ensures that the fitted residuals $\widehat{\eps}_{i,j}(\widehat{\balpha},\widehat{\bbeta})$ are reasonable predictors of the unknown non-observable errors $\eps_{i,j}(\balpha,\bbeta)$. Algorithm~\ref{alg:icls} also provides the fitted residuals as a~side product.

Assumption~C demands a~prior knowledge of parametric copula up to its unknown parameter $\bgamma$. Nevertheless in practical applications, one needs to perform a~copula goodness-of-fit in order to choose a~suitable copula. Assuming that we know the bivariate copula function $C(\cdot,\cdot;\bgamma)$, copula parameter $\bgamma$ is estimated by the \emph{quasi-likelihood} method.

Having in mind that rows of errors $[\eps_{i,1}(\balpha,\bbeta),\ldots,\eps_{i,n+1-i}(\balpha,\bbeta)]$ are independent for all $i$, the full log-likelihood for copula parameter $\bgamma$ with respect to~\eqref{copula-cond} is
\begin{multline*}
\mathcal{L}(\bgamma)=\sum_{i=1}^{n-2}\sum_{j=2}^{n+1-i}\log g_{\balpha,\bbeta}(\eps_{i,j}(\balpha,\bbeta))\\+\sum_{i=1}^{n-2}\sum_{j=3}^{n+1-i}\log c(G_{\balpha,\bbeta}(\eps_{i,j-1}(\balpha,\bbeta)),G_{\balpha,\bbeta}(\eps_{i,j}(\balpha,\bbeta));\bgamma).
\end{multline*}
Ignoring the first term in $\mathcal{L}(\bgamma)$ and replacing $\eps$'s and $G_{\balpha,\bbeta}$ by their estimated counterparts $\widehat{\eps}$'s and $\widehat{G}_n$, parameter $\bgamma$ can be estimated by the so-called canonical maximum likelihood, i.e., maximizing the \emph{partial (pseudo) log-likelihood}:
\begin{align*}
\widehat{\bgamma}&=\arg\max_{\bgamma}\widetilde{\mathcal{L}}(\bgamma),\\
\widetilde{\mathcal{L}}(\bgamma)&=\sum_{i=1}^{n-2}\sum_{j=3}^{n+1-i}\log c(\widehat{G}_n(\widehat{\eps}_{i,j-1}(\widehat{\balpha},\widehat{\bbeta})),\widehat{G}_n(\widehat{\eps}_{i,j}(\widehat{\balpha},\widehat{\bbeta}));\bgamma).
\end{align*}
The correctness of this approach was shown by~\cite{ChenFan2006a}, because the consistency of the canonical likelihood estimate of copula parameter $\bgamma$ was proved under Assumption~C. Here, the unknown unobservable CMV model errors are just replaced by the fitted residuals based on the consistent CMV model parameter estimates. Algorithm~\ref{alg:copula} encapsulates the way of getting copula parameter estimate.



\begin{algorithm}[!htb]
\onehalfspacing\caption{Copula parameter $\bgamma$ estimation by pseudo-likelihood.}
\label{alg:copula}
\algsetup{indent=2em}
\begin{algorithmic}[1]
\REQUIRE Fitted residuals $\{\widehat{\eps}_{i,j}(\widehat{\balpha},\widehat{\bbeta})\}_{i=1,j=2}^{n-1,n+1-i}$ and copula density $c(\cdot,\cdot;\bgamma)$.
\ENSURE Copula parameter estimate $\widehat{\bgamma}$.
\STATE marginal ecdf $\widehat{G}_n(e) \leftarrow \frac{1}{n(n-1)/2+1}\sum_{i=1}^{n-1}\sum_{j=2}^{n+1-i}\mathcal{I}\{\widehat{\eps}_{i,j}(\widehat{\balpha},\widehat{\bbeta})\leq e\}$
\STATE $\widehat{\bgamma} \leftarrow \arg\max_{\bgamma}\sum_{i=1}^{n-2}\sum_{j=3}^{n+1-i}\log c(\widehat{G}_n(\widehat{\eps}_{i,j-1}(\widehat{\balpha},\widehat{\bbeta})),\widehat{G}_n(\widehat{\eps}_{i,j}(\widehat{\balpha},\widehat{\bbeta}));\bgamma)$
\end{algorithmic}
\end{algorithm}

\cite{ChenFan2006a} also remarked that the empirical distribution function $\widehat{G}_n$ can be \emph{smoothed by kernels} as an alternative estimate of the marginal distribution of errors. This can especially be helpful in case of a~smaller number of residuals available.

Now, the whole estimation procedure became a~\emph{semiparametric} one.

\section{Prediction of Reserves and Estimating Their Distribution}\label{sec_prediction}
The main goals in actuarial reserving are prediction of reserves $R_i^{(n)}$ and, consequently, estimation of the reserves' distribution, e.g., in order to obtain quantiles---$99.5\%$ quantile for the Solvency II purposes (VaR calculation).

A~predictor for reserve $R_i^{(n)}$ can be defined as
\[
\widehat{R}_i^{(n)}=\widehat{Y}_{i,n}-Y_{i,n+1-i}.
\]
Therefore, finding a~predictor $\widehat{Y}_{i,n}$ for $Y_{i,n}$ is crucial. One naive proposal is to predict $Y_{i,j}$ as conditional mean of previous claim amount $Y_{i,j-1}$ with the plugged-in CMV model parameter estimates
\begin{equation}\label{wrong_pred}
\widecheck{Y}_{i,j}:=\E_{\balpha,\bbeta}[Y_{i,j}|Y_{i,n+1-i}]|_{\balpha=\widehat{\balpha},\bbeta=\widehat{\bbeta}},\quad i=2,\ldots,n,\,j=n+2-i,\ldots,n.
\end{equation}
This approach gives $\widecheck{Y}_{i,j}=\mu(\widecheck{Y}_{i,j-1},\widehat{\balpha},j),\, i+j> n+1$. 
In spite of that, such approach would be eligible only if prediction~\eqref{wrong_pred} was unbiased or at least asymptotically unbiased, which is not justified.

\subsection{Semiparametric Bootstrap}
Prediction of unobserved claims may be done in a~\emph{telescopic} way based on the CMV model formulation: start with the diagonal element $Y_{i,n+1-i}$ and predict $Y_{i,j},\,j>n+1-i$ stepwise in each row
\begin{subequations}\label{prediction}
\begin{align}
\widehat{Y}_{i,j}&=Y_{i,j},\quad i+j\leq n+1;\\
\widehat{Y}_{i,j}&=\mu(\widehat{Y}_{i,j-1},\widehat{\balpha},j)+\sigma(\widehat{Y}_{i,j-1},\widehat{\bbeta},j)\widetilde{\eps}_j,\quad i+j> n+1.
\end{align}
\end{subequations}
Errors $\widetilde{\eps}_j$ are simulated from the fitted residuals. One convenient resampling procedure is the semiparametric bootstrap which takes advantage of the fact that $\eps_{i,j}(\balpha,\bbeta)=G_{\balpha,\bbeta}^{-1}(X_{j})$ for all $i$ (due to the independent rows), where $\{X_j\}_{j=2}^n$ is a~stationary first-order Markov process with the copula $C(x_1,x_2;\bgamma)$ being the joint distribution of $[X_{j-1},X_{j}]$.

When errors $\widetilde{\eps}_j$ are simulated sufficiently many times, the empirical (bootstrap) distribution of $\widehat{Y}_{i,n}$ is obtained, which should mimic the true unknown distribution of $Y_{i,n}$. Thereafter, an~estimate of reserves' distribution is acquired and some imported quantities of the reserves can be easily calculated, e.g., mean, variance, or quantiles of the reserves. The whole procedure is summed up in Algorithm~\ref{alg:pred}.

\begin{algorithm}[!htb]
\onehalfspacing\caption{Predictions of claims reserves and estimation of their distributions by semiparametric bootstrap.}
\label{alg:pred}
\algsetup{indent=2em}
\begin{algorithmic}[1]
\REQUIRE Latest diagonal of cumulative claims $\{Y_{i,n+1-i}\}_{i=2}^{n}$, empirical distribution function $\widehat{G}_n$ of residuals, parametric mean and variance functions $\mu$ and~$\sigma$, inverse of conditional parametric copula $C^{-1}_{2|1}$, parameter estimates $\widehat{\balpha}$, $\widehat{\bbeta}$, and $\widehat{\bgamma}$, number of bootstrap simulations $B$.
\ENSURE Empirical bootstrap distribution of predicted reserves $\widehat{R}_i^{(n)}$, i.e., the empirical distribution where probability mass $1/B$ concentrates at each of ${}_{(1)}\widehat{R}_i^{(n)},\ldots,{}_{(B)}\widehat{R}_i^{(n)}$, $i=2,\ldots,n$.
\FOR{$b=1$ \TO $B$}
\STATE generate $n-1$ independent $U(0,1)$ random variables $\{X_j\}_{j=2}^{n}$
\STATE ${}_{(b)}U_2 \leftarrow X_2$
\STATE ${}_{(b)}\widehat{\eps}_2 \leftarrow \widehat{G}_n^{-}({}_{(b)}U_2)$ \COMMENT{$\widehat{G}_n^{-}$ is the empirical quantile function}
\FOR{$j=3$ \TO $n$}
\STATE ${}_{(b)}U_j \leftarrow C^{-1}_{2|1}(X_j|{}_{(b)}U_{j-1};\widehat{\bgamma})$
\STATE ${}_{(b)}\widehat{\eps}_j \leftarrow \widehat{G}_n^{-}({}_{(b)}U_j)$
\ENDFOR
\ENDFOR
\STATE center bootstrap residuals ${}_{(b)}\widetilde{\eps}_j \leftarrow {}_{(b)}\widehat{\eps}_j - \frac{1}{n-1}\sum_{l=2}^n {}_{(b)}\widehat{\eps}_l$, $b=1,\ldots,B$
\FOR{$i=2$ \TO $n$}
\FOR[repeat to obtain empirical distribution of $\widehat{R}_i^{(n)}$]{$b=1$ to $B$}
\STATE ${}_{(b)}\widehat{Y}_{i,n+1-j} \leftarrow Y_{i,n+1-i}$
\FOR{$j=n+2-i$ \TO $n$}
\STATE ${}_{(b)}\widehat{Y}_{i,j} \leftarrow \mu({}_{(b)}\widehat{Y}_{i,j-1},\widehat{\balpha},j)+\sigma({}_{(b)}\widehat{Y}_{i,j-1},\widehat{\bbeta},j){}_{(b)}\widetilde{\eps}_j$
\ENDFOR
\STATE ${}_{(b)}\widehat{R}_i^{(n)} \leftarrow {}_{(b)}\widehat{Y}_{i,n}-Y_{i,n+1-i}$
\ENDFOR
\ENDFOR
\end{algorithmic}
\end{algorithm}



%

\section{Real Data Analysis}\label{sec_data}
The proposed method is illustrated on a~data set from~\cite{Zehnwirth}. Here, data for 11 accident years are available ($n=11$). The data analyses are conducted in R~program~\citep{Rko}.

Several candidates for the mean and variance function $\mu$ and $\sigma$ were considered (see Section~\ref{sec_model}). The CLS were used to estimate the parameters of the CMV model. By inspecting the fitted residuals graphically and comparing the residuals' mean square error, the following structure of the CMV model was taken into account
\begin{subequations}\label{CMVstructure}
\begin{align}
\mu(Y_{i,j-1},\balpha,j)&=\left(1+\alpha_1\alpha_2j^{-1-\alpha_2}\exp\left\{\alpha_1j^{-\alpha_2}\right\}\right)Y_{i,j-1},\\
\sigma(Y_{i,j-1},\bbeta,j)&=\beta_1\exp\{-\beta_2j\}\sqrt{Y_{i,j-1}}.
\end{align}
\end{subequations}
The CLS estimates are $\widehat{\alpha}_1=2.033$, $\widehat{\alpha}_2=1.106$, $\widehat{\beta}_1=109.8$, and $\widehat{\beta}_2=0.4053$. Figure~\ref{fig:residuals} shows the fitted residuals for the CMV model structure~\eqref{CMVstructure}.

\begin{figure}[!htb]
\centering
\includegraphics[width=\textwidth]{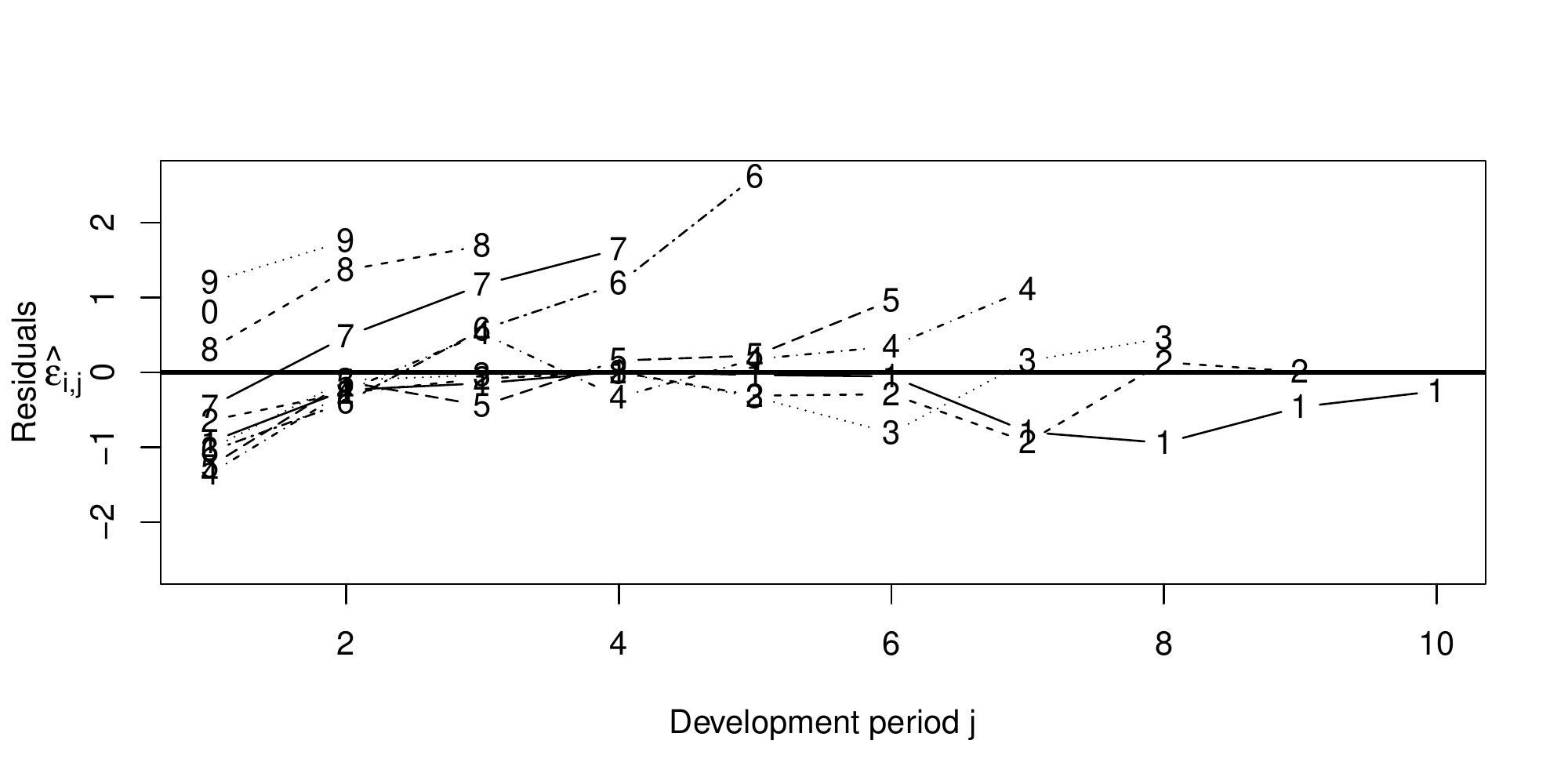}
\caption{Residuals of the CMV model~\eqref{CMVstructure} with a~common accident year are connected by lines ($0$ stands for the accident year $i=10$).}\label{fig:residuals}
\end{figure}

There is still some slight pattern (trend) not captured by mean and variance parametric part of the model (Figure~\ref{fig:residuals}), but that will be modeled by the residuals' dependence. Kendall $\tau$ for the pairs of consecutive residuals $\{[\widehat{\eps}_{i,j-1}(\widehat{\balpha},\widehat{\bbeta}),\widehat{\eps}_{i,j}(\widehat{\balpha},\widehat{\bbeta})]\}_{i=1,j=3}^{n-2,n+1-i}$ equals $0.43$, which indicates at least mild dependence.

Three Archimedean copulae (Clayton, Frank, and Gumbel) together with Gaussian and Student $t_5$-copula (the degrees of freedom were set in order to have finite fourth moment) were considered for modeling the dependence among the errors. We have performed the $S_n^{(C)}$ goodness-of-fit test proposed by~\citet[Section~4]{GRB2009}. This test is preferred in our case for two reasons: it was regarded by the above cited power study as a~recommendable test and it relies on Rosenblatt's transform (a~single bootstrap is enough to approximate the null distribution of the test statistic and extract $p$-values). The \emph{Gumbel copula} ($\widehat{\bgamma}=1.776$) was chosen according to the goodness-of-fit test (providing the highest $p$-value equal $0.30$). The Gumbel copula exhibits strong right tail dependence and relatively weak left tail dependence. In our data set, the transformed residuals---transformed by the residuals' marginal empirical distribution function $\widehat{G}_n$ (having uniform margins)---seem to be strongly correlated at high values but less correlated at low values (see Figure~\ref{fig:copula-gumbel}). Then, the Gumbel copula is indeed an appropriate choice. The pairs of original residuals $[\widehat{\eps}_{i,j-1}(\widehat{\balpha},\widehat{\bbeta}),\widehat{\eps}_{i,j}(\widehat{\balpha},\widehat{\bbeta})]$ are shown in Figure~\ref{fig:copula-gumbel} as well. Moreover, Figure~\ref{fig:copula-gumbel} contains the density for the Gumbel copula $C(\cdot,\cdot;\widehat{\bgamma})$ and the first $500$ pairs of the bootstrapped residuals, which should jointly ``behave'' like the original ones.


\begin{figure}[!htb]
\centering
\includegraphics[width=\textwidth]{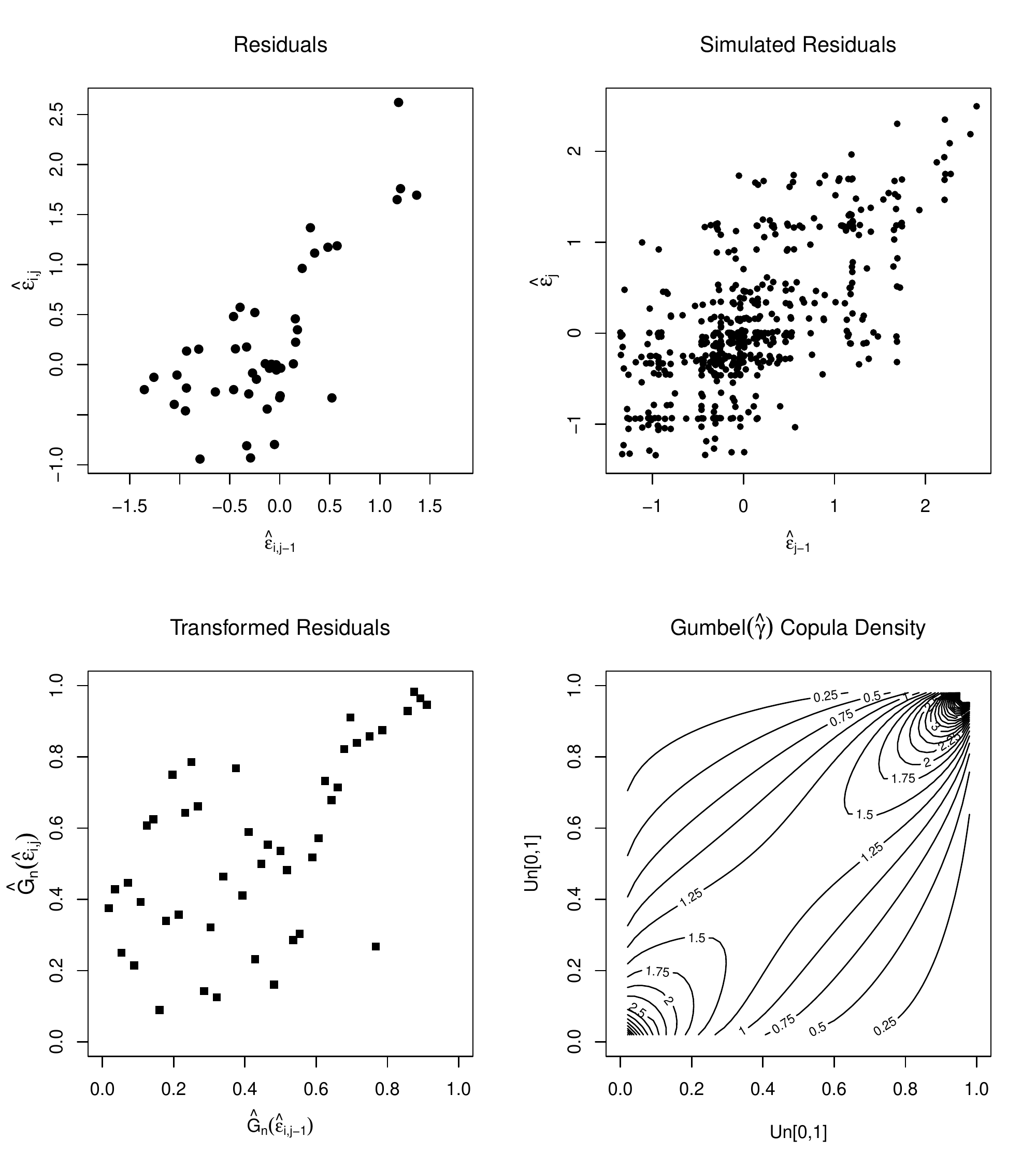}
\caption{Pairs of original fitted residuals $[\widehat{\eps}_{i,j-1}(\widehat{\balpha},\widehat{\bbeta}),\widehat{\eps}_{i,j}(\widehat{\balpha},\widehat{\bbeta})]$ (top left), transformed residuals by the marginal empirical distribution function $\widehat{G}_n$ (bottom left), estimated density $c(\cdot,\cdot;\widehat{\bgamma})$ of the Gumbel copula (bottom right), simulated (bootstrapped) residuals $[{}_{(b)}\widehat{\eps}_{j-1},{}_{(b)}\widehat{\eps}_{j}]$ from the Gumbel copula (top right).}\label{fig:copula-gumbel}
\end{figure}

As a~benchmark to the CMV model with copula Assumption~C, the classical bootstrap version of chain ladder~\citep[A3.1]{england2002} is chosen. It has some disadvantages, which can be overcome by our approach: the number of parameters (development factors) depending on the sample size, some parameters (the last development factor) estimated by just ratio of two numbers (yielding zero sample variance), questionable consistency of the estimates (conditional consistency recently proved by~\cite{PH2012}), and the non-realistic assumption of independence of the residuals.


The results of the proposed approach---\emph{conditional least squares with copula (CLSC)}---and the traditional one---\emph{bootstrapped chain ladder (BCL)}---are compared numerically in Table~\ref{tab:compar} and graphically in Figure~\ref{fig:reserves}. In both cases, the number of bootstrap replications is $B=5000$.

\begin{table}[!htb]
\renewcommand{\arraystretch}{1.1}
\begin{center}
\begin{tabular}{cx{1.1cm}x{1.1cm}x{1.1cm}x{1.1cm}x{1.1cm}x{1.1cm}x{1.1cm}x{1.1cm}}
\toprule
$i$ & \multicolumn{2}{c}{Reserve} & \multicolumn{2}{c}{Standard error} & \multicolumn{2}{c}{$95\%$-quantile} & \multicolumn{2}{c}{$99.5\%$-quantile} \tabularnewline
\cmidrule(lr){2-3} \cmidrule(lr){4-5} \cmidrule(lr){6-7} \cmidrule(lr){8-9}
 & BCL & CLSC & BCL & CLSC & BCL & CLSC & BCL & CLSC\tabularnewline
\midrule
2 & $14$ & $15$ & $5$ & $1$ & $24$ & $16$ & $31$ & $17$\tabularnewline
3 & $38$ & $38$ & $8$ & $2$ & $51$ & $42$ & $61$ & $44$\tabularnewline
4 & $64$ & $65$ & $10$ & $3$ & $81$ & $70$ & $91$ & $73$\tabularnewline
5 & $100$ & $102$ & $12$ & $4$ & $121$ & $109$ & $132$ & $114$\tabularnewline
6 & $144$ & $145$ & $14$ & $6$ & $168$ & $156$ & $183$ & $162$\tabularnewline
7 & $211$ & $208$ & $17$ & $8$ & $241$ & $222$ & $261$ & $231$\tabularnewline
8 & $385$ & $373$ & $24$ & $13$ & $427$ & $396$ & $450$ & $411$\tabularnewline
9 & $765$ & $728$ & $38$ & $24$ & $831$ & $769$ & $868$ & $794$\tabularnewline
10 & $1\,362$ & $1\,294$ & $59$ & $44$ & $1\,464$ & $1\,370$ & $1\,531$ & $1\,413$\tabularnewline
11 & $2\,195$ & $2\,153$ & $111$ & $81$ & $2\,389$ & $2\,293$ & $2\,508$ & $2\,374$\tabularnewline
\midrule
Total & \cellcolor[gray]{0.6}{$5\,279$} & \cellcolor[gray]{0.3}{\color{white} $5\,122$} & $172$ & $115$ & $5\,565$ & $5\,317$ & \cellcolor[gray]{0.6}{$5\,751$} & \cellcolor[gray]{0.3}{\color{white} $5\,432$}\tabularnewline
\bottomrule
\end{tabular}
\end{center}
\caption{BCL and CLSC reserving results for~\cite{Zehnwirth} data (values in thousands).}
\label{tab:compar}
\renewcommand{\arraystretch}{1.0}
\end{table}

\begin{figure}[!htb]
\centering
\includegraphics[width=.86\textwidth]{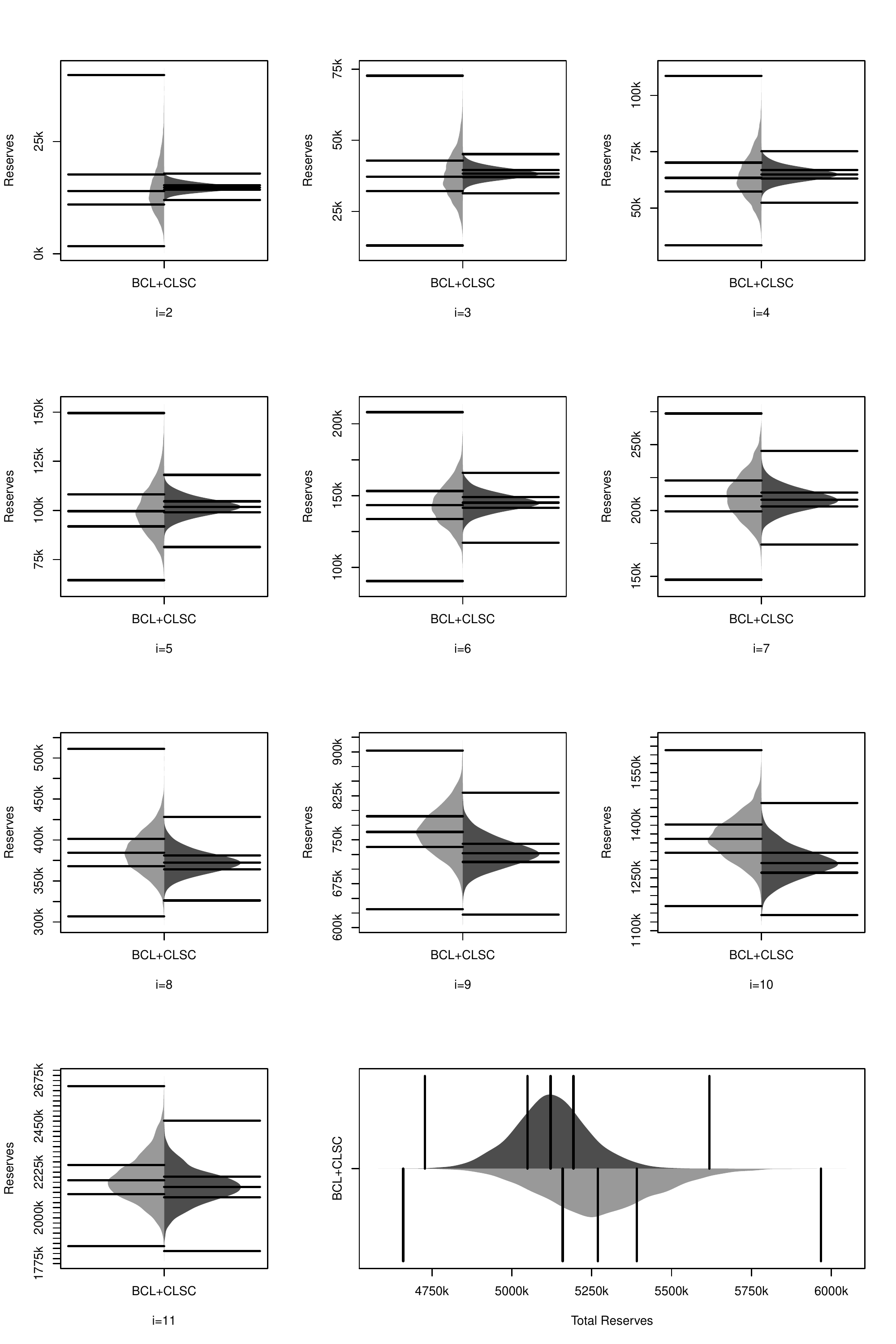}
\caption{Reserves distribution based on BCL (light gray) and CLSC (dark gray).}\label{fig:reserves}
\end{figure}

For each particular accident year and in total as well, the CLSC approach provides slightly smaller predictions of reserves than the BCL one. But even more important is that the estimates of the reserves' distribution are \emph{less volatile} (e.g., smaller standard error or smaller difference between the third and the first quartile) and have \emph{lighter right tails} (and, hence, smaller high quantiles) for the CLSC compared to the BCL. By virtue of such unambiguous outcomes, wasting number of parameters and impetuous neglecting of present dependence in data really matters.

Additionally, an interesting finding comes out, when the CMV model with copula approach is compared to the approach suggested by~\cite{HP2013}, where the dependence is modeled by generalized estimating equations (GEE). Such a~different method provides very similar results. For instance, the GEE (with logarithmic link function, AR(1) correlation structure, and linear variance function) provides coefficient of variation (CoV) $1.95$ for the~\cite{Zehnwirth} data, whereas the CoV in case of the CLSC method is equal to $1.99$. For a~complete comparison, the CoV for the traditional nonparametric BCL, which disregards dependence, equals  $3.29$.

\section{Conclusions and Discussion}
This paper proposes the \emph{conditional mean and variance (CMV)} time series model with innovations being a~\emph{stationary first-order Markov process}. Such a~framework is demonstrated to be suitable for stochastic claims reserving in general insurance. It brings several advantages, which should be though of a~natural relaxation of some too restrictive assumptions in the traditional methods. Indeed in contrast to the classical reserving techniques, a~large number of the mean and variance structures for claims is allowed by the general model's definition yielding a~very \emph{flexible modeling approach}, relatively \emph{smaller number of model parameters} not depending on the number of development periods is required, and time series \emph{innovations} (errors) are \emph{not} considered as \emph{independent}. All of the previously mentioned benefits contribute to the increase in precision of the claims reserves' prediction (e.g., smaller variability, lighter tails). Furthermore, \emph{claims forecasting} beyond the last observed development period becomes straightforward.

The \emph{conditional least squares (CLS)} together with the \emph{copula} approach provide parameter estimates of the assumed model. Moreover, \emph{semiparametric bootstrap}, as an~extension based on the CLS and copulae modeling purview, is used in order to estimate the whole distribution of predicted reserves.

Generally, the consistency of every bootstrap procedure relies on the \emph{consistency of the parameter estimates}. In the proposed approach, the consistency of the CLS parameter estimates in the CMV model is shown, which gives \emph{validity for the bootstrapping}. This brings the CMV model to the fore, because the consistency for the chain ladder development factors (together with the necessary and sufficient condition), that has recently been proved by~\cite{PH2012}, is only conditional, which makes the further usage more complicated.

Our approach can also be robustified in the sense that not only the last observed diagonal of claims is taken into account for the reserves' prediction as in~\eqref{prediction}, but the telescopic prediction begins from each of the observed claims. Finally, the procedure of predicting reserves can even be generalized by incorporating time-varying copulae, which enables to omit the stationary dependence structure, although, it requires more complicated setup.

\section*{Acknowledgments}
The research of Michal Pe\v{s}ta was supported by the Czech Science Foundation project GA\v{C}R No.~P201/13/12994P. Financial support through CRC 649 ``\"Okonomisches Risiko'' is gratefully acknowledged by Ostap Okhrin.

\bibliographystyle{elsarticle-harv}
\bibliography{claims}







\end{document}